\documentclass[10pt,oneside,a4paper,reqno]{amsart}  
\usepackage{mathrsfs}
\usepackage{tikz-cd}
\usetikzlibrary{cd}
\usepackage{amssymb,
enumerate}
\usepackage{bbm}
\usepackage{mathtools}
\usepackage{esint}
\usepackage{geometry}                
\usepackage[T1]{fontenc}
\usepackage[english]{babel}
\usepackage{color}
\usepackage[colorlinks=true, linkcolor=magenta, citecolor=cyan, urlcolor=blue]{hyperref}

\renewcommand{\phi}{\varphi}

\newcommand{\OOO}[1]{O \left(#1\right)}

\def\P{\mathbb P}
\def\be{\begin{equation}}
\def\ee{\end{equation}}
\def\bea{\begin{eqnarray}}
\def\eea{\end{eqnarray}}
\def\ni{\noindent}
\def\nn{\nonumber}

\DeclareMathOperator{\diag}{diag}

\def\ie{\textit{i.e. }}

\def\a{\alpha}
\def\e{\varepsilon}
\def\d{\delta}
\def\g{\gamma}

\def\b{\beta}
\def\D{\Delta}
\def\l{\lambda}
\def\r{\rho}
\def\s{\sigma}

\def\R{\mathbb{R}}

\def\N{\mathbb{N}}

\theoremstyle{plain}
\newtheorem{theorem}{Theorem}[section]
\newtheorem{lemma}[theorem]{Lemma}

\newtheorem{corollary}[theorem]{Corollary}

\numberwithin{equation}{section}

\definecolor{light}{gray}{.9}

\author{Giuseppe Genovese}
\address{Giuseppe Genovese: Department of Mathematics and Computer Science, University of Basel, Spiegelgasse 1, 4051 Basel, Switzerland.}
\email{giuseppe.genovese@unibas.ch}

\author{Daniele Tantari}
\address{Daniele Tantari: University of Bologna, Mathematics Department, Via Zamboni, 33, 40126, Bologna (Italy).}
\email{daniele.tantari@unibo.it}

\title[Legendre equivalences of spherical Boltzmann machines]
{Legendre equivalences of spherical Boltzmann machines}
\date{\today}

\begin{document}

\begin{abstract}
We study either fully visible and restricted Boltzmann machines with sub-Gaussian random weights and spherical or Gaussian priors. We prove that the free energies of the spherical and Gaussian models are related by a Legendre transformation. Incidentally our analysis brings also a new purely variational derivation of the free energy of the spherical models. 
 
%
\end{abstract}

\pagestyle{plain}

\maketitle

\section{Introduction}

Originally inspired by statistical physics \cite{hop}, Boltzmann machines (BMs) \cite{BM,rm} are among the most studied data generative models, playing a central role in the phenomenal progresses of machine learning through neural networks of the last two decades. In particular restricted BMs (RBMs) constitute a cornerstone of unsupervised learning, mainly for the very successful training algorithms developed \cite{CD, CDp}, working also for many interesting deep architectures \cite{DBN, auto, belief}, for which RBMs are used as the basic building blocks \cite{deepBM, deepbook}. 

Concretely a BM is a probability distribution of the Gibbs type which is aimed to reproduce the true distribution of the data. In the much useful neural network interpretation the units of the machine should mirror the data, that is typical configurations according to the BM distribution are desired to be close to typical data. Therefore two ingredients are crucial to build up a BM: the energy function and the a priori unit distribution. The main focus of the paper will be on fully visible BMs, namely Hopfiel models, and RBMs with spherically symmetric priors. We will give a formula for the free energy of these BMs pointing out a Legendre duality between rigid spherical priors and a certain quite general class of sub-Gaussian distributions, already investigated for the Sherrington-Kirkpatrick energy \cite{legendre}. This latter equivalence is achieved by a suitable adaption of a very much established method from the statistical physics tradition, namely equivalence of ensemble (spherical and Gaussian). 

\

\subsection{Set up}

First we will introduce the models we will deal with. We start by the energy function.

Let $\{\xi_{ij}\}_{i=1,\ldots,N_1\,,j=1,\ldots,N_2}$ a doubly indexed sequence of i.i.d. centred sub-Gaussian r.vs. with 
$$
E[\xi_{ij}\xi_{hk}]=\d_{ih}\d_{jk}\,, 
$$
For definiteness we  assume that
$$
\frac{N_1}{N_1+N_2}\to \a\in(0,1)\,.
$$
We shall look at this sequence in two different ways, namely as entries of a $N_1\times N_2$ random matrix $\Xi$ or a collection of $N_2$ patterns in $\R^{N_1}$, defining a sample covariance matrix  $\frac{1}{N_2}\Xi\Xi^T\in\R^{N_1\times N_1}$. In either cases we can use the following two important properties of rectangular random matrices with centred independent subgaussian entries with unitary variance \cite{bai,versh}. For $A\in\R^{N\times N}$ we denote its eigenvalues as $\l_i:=\l_{i}(A)$, $i=1,\ldots N$; for $A\in\R^{N_1\times N_2}$ we denote its singular values as $\s_i:=\s_{i}(A)$, $i=1,\ldots N_1$.

\

\begin{itemize}
\item[\bf P1)] The empirical distribution of the eigenvalues of $\frac{1}{N_2}\Xi\Xi^T$ converges a.s. to the Marchenko-Pastur law:
$$
\frac{1}{N_1}\sum_{j=1}^{N_1}\d(\l-\l_i(\frac{1}{N_2}\Xi\Xi^T)) \rightharpoonup \r_{MP}(\l;\a)\quad \P-a.s.
$$
where
\be\label{eq:MP}
\r_{MP}(\l;\a):= (2-\frac{1}{\a})^+\d_0(\l) +\frac{(1-\alpha)}{2\pi \alpha}\frac{\sqrt{(\l-\l_-)(\l_+-\l)}}{\l}1_{[\l_-,\l_+]}(\l)
\ee
with  $\l_{\pm}:=\left(1\pm\sqrt{\a/(1-\a)}\right)^2$. 

\item[\bf P2)] The spectrum of $\frac{1}{N_2}\Xi\Xi^T$ is localised in an interval with large probability:

\be\label{eq:impo-spectrum}
\P(\|\frac {1}{N_2}\Xi\Xi^T-\mathbb I\|_{op}\geq t+\sqrt{\l_+}-1)\leq 2e^{-\frac {N_1t^2}{2}}\,. 
\ee
\end{itemize}

\

We will deal with two kind of Boltzmann machines: Hopfield models (HMs) and restricted Boltzmann machines (RBMs). Their energy functions (or Hamiltonians) are 
\bea
H_{N_1,N_2}^{HM}&:=&-\frac{1}{N_1+N_2}\sum_{j=1}^{N_2}\sum_{i<k}^{N_1} \xi_{ij}\xi_{kj}x_ix_k
\,,\label{eq:HM}\\
H_{N_1,N_2}^{RBM}&:=&-\frac{1}{\sqrt{N_1+N_2}}\sum_{j=1}^{N_2}\sum_{i=1}^{N_1} \xi_{ij}x_iy_j\,\label{eq:RBM}
\eea
(we will often drop all the indexes from the energies to lighten the notations).

\

In the Hopfield model units have one single choice for the prior distribution, while a RBM is an undirected bipartite system in which we can have different priors for each layer. With this in mind, we can now introduce the prior distributions we shall deal with. Let $S^{N}(R)$ be the $(N-1)$-dimensional sphere in $\R^N$ with radius $R\sqrt N$. Define the following a priori probability measures on $\R^N$:
\be\label{eq:prior}
\begin{cases}
\s_{N,R}(dx)&\mbox{uniform on $S^N(R)$}\,;\\
\g_{N,\theta}(dx)&\mbox{centred Gaussian with covariance $\theta \mathbb I$}\,.
\end{cases}
\ee
Models with Gaussian priors are typically ill-defined for low temperatures and need a sub-Gaussian regularisation. Let  $r\,:\R\mapsto\R$ such that there is $\e>0$ with
$$
\lim_{x\to\infty}\frac{r(x)}{x^{2+\e}}=\infty\,.
$$
We define the regularised prior on $\R^N$
\be\label{eq:Reg-prior1}
\r_N(dx):= e^{-Nr\left(\frac{\|x\|}{\sqrt N}\right)}\g_{N,\theta}(dx)\,. 
\ee
For instance in \cite{BDG} it was considered $r(x)=\b x^4/4$ while in \cite{gauss,legendre} $r(x)=\b x^4/4-\l x^2/2$. 

This notion extends easily to two layer settings. We introduce some $r\,:\,\R^2\mapsto \R$ so that
$$
\lim_{x^2+y^2\to\infty} \frac{r(x,y)}{x^2}=\lim_{x^2+y^2\to\infty} \frac{r(x,y)}{y^2}=\infty\,.
$$
Then we define the following measure on $\R^{N_1}\times \R^{N_2}$
\be\label{eq:Reg-prior2}
\r^2_{N_1,N_2}(dxdy):= e^{-\sqrt{N_1N_2}r\left(\frac{\|x\|}{\sqrt N_1},\frac{\|y\|}{\sqrt N_2}\right)}\g_{N_1,\theta}(dx)\g_{N_2,\theta}(dy)\,. 
\ee
The idea is that $\r$ can be used to regularise a single layer, while $\r^2$ regularises two layers at once. A simple example is low rank matrix factorisation with Gaussian priors in which $r(x,y)=x^2y^2/2$ \cite{leo}. 

One technical problem is that the support of $\s_{N,R}$ has zero $\g_{N,\theta}$-measure, that is $\g_{N,\theta}(S^{N}(R))=0$. For this reason for a given $\e>0$ we need to introduce the spherical shells
\be\label{eq:Sfshell}
S_{N,R,\e}:=\{x\in\R^N\,:\,R-\e\leq\|x\|_2\leq R+\e\}\,.
\ee
We denote (here and further $|A|$ is the Lebesgue measure of the set $A$)
$$
\s_{N,R.\e}(dx):=|S_{\e,N,R}|^{-1}1_{\{S_{\e,N,R}\}}(x)dx\,
$$
the uniform probability on a shell and note that it is a.c. w.r.t. $\g_{N,\theta}$. We will also often make use of the simple relation
\be\label{eq:shells-surface}
|S_{\e,N,R}|=\e|S_{N,R}|+\OOO{\e^2}\,.
\ee

With these definitions at hand, we can introduce the probability distributions defining our Boltzmann machines. Let $\mu_{(\cdot)}\in\{\s_{(\cdot)},\g_{(\cdot)},\r_{(\cdot)}\}$ and $\mu^2_{(\cdot)}\in\{\s_{(\cdot)},\g_{(\cdot)}\}\otimes \{\s_{(\cdot)},\g_{(\cdot)}\}\cup\{\r_{(\cdot)}^2\}$ denote one prior among the one introduced before respectively for the one-layer and the bipartite machine. We have for $\b>0$
\be\label{eq:Gibbs}
G^{HM}_{N_1,N_2,\b}(dx):=Z_{N_1,N_2,\b}^{-1}\mu_{N_1}(dx)e^{-\b H(x)}\,\qquad G^{RBM}_{N_1,N_2,\b}:=Z_{N_1,N_2,\b}^{-1}\mu^2_{N_1,N_2}(dxdy)e^{-\b H(x,y)}\,.
\ee
The normalisation $Z_{N_1,N_2,\b}$ is called partition function and it needs not to be the same, despite the symbol. Moreover
\be\label{eq:free}
A^{HM}_{N_1,N_2}(\b):=\frac{1}{N_1+N_2}\log Z_{N_1,N_2,\b}\qquad A^{RBM}_{N_1,N_2}(\b):=\frac{1}{N_1+N_2}\log Z_{N_1,N_2,\b}\,. 
\ee

Most interesting is to evaluate the last quantities in the so-called thermodynamic limit, namely
$$
N_1,N_2\to\infty\,,\quad\mbox{such that }\qquad \lim_{N_1,N_2}\frac{N_1}{N_1+N_2}=:\a>0\,.
$$ 
The regime $\a\neq0,1$ is called high-load and we will stick to that in this work. 

Being Lipschitz functions of the weights, free energies always satisfies a concentration inequality which ensures their a.s. convergence to the expected value. This self-averaging property will be exploited throughout without further mention. 

For simplicity we will assume always the distributions of the two layers to have the same parameters, \ie radius and variance. The general case requires a trivial extension of our formulas. 

In general $A_\mu(X,\a,\b)$ denotes the free energy of a BM with prior $\mu$ whose parameters are $X$. We will often drop the descriptive labels from the Hamiltonian, Gibbs measure and free energy when it will be clear from the context to which one we refer to. Only exception, the quantities of interest referred to spherical shell priors are indicated by $\widehat \cdot$ for all the models.

\

\subsection{Main result and organisation of the paper}

We will focus on the free energy associated to BMs with the particular priors introduced above. 
We will prove the following equivalence at the level of free energies, which can be related by a marginalisation (m) or a Legendre transform (LT): 

\

\begin{center}
\begin{tikzcd}
  HM_\s \arrow[r,leftrightarrow, "m", blue] \arrow[d,leftrightarrow,"LT", blue]
    & RBM_{\s,\g} \arrow[d,leftrightarrow,"LT", blue] \arrow[r,leftrightarrow, "LT",red]&RBM_{\s\s}\arrow[r,leftrightarrow, "LT",red]&RBM_{\r^2}\\
  HM_{\r} \arrow[r, leftrightarrow, "m", blue]
&RBM_{\r\g} \arrow[ru,leftrightarrow, "LT",red] \end{tikzcd}
\end{center}

We will not concern here about low-load ($\a=0,1$), yet some of the equivalences we state hold also in this regime. More precisely, red arrows indicate equivalences valid only in high load while blue arrow equivalences hold regardless of $\a$.

Marginalisation is the usual trick of RBMs. The two layers are coupled linearly, so that one can integrate out \'a la Stratonovich the units from the Gaussian layer in the partition function of a RBM to obtain the partition function of a HM (with $\b^2$ replacing $\b$). 
All the relevant quantities (e.g. Gibbs measures, free energy, order parameters etc) of one model can be computed directly from the one of the other one. For instance the equivalence $HM_\s\leftrightarrow RBM_{\s,\g}$ at the level of Gibbs measure reads as
$$
Z^{-1}_{N_1,N_2,\b}\s_{N_2,R}(dx)\int_{\R^{N_2}}\g_{N_2,\theta}(dy)e^{\b H_{N_1,N_2}^{RBM}}=Z^{-1}_{N_1,N_2,\b}\s_{N_1,R}(dx)e^{\frac{\b^2}{\theta^2} H^{HM}_{N_1}}\,,
$$
since the  partition functions of the two models are numerically the same.

Legendre transforms are where the idea of equivalence of ensembles exploits and more precise statements are given in Theorem \ref{th:main} below. A Gaussian prior of $N$ units concentrates on a N-dimensional sphere of radius proportional to $\sqrt N$. To find the optimal radius we slice up the Gibbs measure at the level of the Gaussian prior and look for the most relevant contribution to the free energy. This strategy yields naturally a variational principle of the Legendre type relating the spherical and Gaussian free energies. Moreover we identify the square radius of the optimal sphere $R^2$ and the variance of the Gaussian model $\theta^{-1}$ as Legendre conjugate variables. 

The main idea is very simple; we briefly outline the heuristics for the equivalence $RBM_{\s\s}\xleftrightarrow{LT} RBM_{\r^2}$ (the other cases are similar). By the standard disintegration of finite-dimensional Gaussian measures into spheres we have
\bea
&&\int_{\R^{N_1}\times \R^{N_1} }\mu^2(dxdy)e^{-\b H(x,y)}\nn\\
&=&\int_{0}^{\infty} dR_1\int_{0}^{\infty}dR_2 e^{-\sqrt{N_1N_2}r\left(\frac{R_1}{\sqrt{N_1}},\frac{R_2}{\sqrt{N_2}}\right)-\frac{R_1}{2\theta_1}-\frac{R_2}{2\theta_2}}\nn\\
&&\frac{|S_{N_1}(R_1)\times S_{N_2}(R_2)|}{\sqrt{2\pi\theta}^{N_1+N_2}}\int_{S_{N_1}(R_1)\times S_{N_2}(R_2)}\s_{N_1,R_1}(dx)\s_{N_2,R_2}(dy)e^{-\b H(x,y)}\,.\label{eq:cont}
\eea
We adjusted the normalisation of the inner integral so to get the partition function of $RBM_{\s\s}$. Thus we continue the chain of identities as
\be
\eqref{eq:cont}=\int_{0}^{\infty} dR_1\int_{0}^{\infty}dR_2 e^{-\sqrt{N_1N_2}r\left(\frac{R_1}{\sqrt{N_1}},\frac{R_2}{\sqrt{N_2}}\right)-\frac{R_1}{2\theta_1}-\frac{R_2}{2\theta_2}+\log\left(\frac{|S_{N_1}(R_1)\times S_{N_2}(R_2)|}{\sqrt{2\pi\theta}^{N_1+N_2}}\right)+(N_1+N_2)A_{\s\s}(R_1,R_2)}\,. 
\ee
Then we can evaluate the integral by the usual Laplace method, since by a simple scaling argument the maximum must be attained at the scale $R_1,R_2\sim\sqrt{N_1+N_2}$. 

To give rigorous grounds to this heuristics we need some few properties. First, we have to control the thermodynamic limit of the free energy of spherical models $HM_\s$ and $RBM_{\s\s}$. These limits are computed in Section \ref{section:sferico}. Secondly, to properly implement the above disintegration formula, we need the regular behaviour of the model with spherical shell prior as the thickness of the shell vanish. In other words, at the level of free energy thermodynamics should favour those configurations on the shell which actually lie on a given sphere into it. This is proven in sections \ref{section:sferico} and \ref{section:legendre}. Lastly, the unit configurations outside any ball of radius growing faster  than $\sqrt N_1, \sqrt N_2$ must give vanishing contribution in the thermodynamic limit. This is proven in Section \ref{section:legendre}, where the proof of our main result, \ie subsequent Theorem \ref{th:main}, is completed.

\begin{theorem}\label{th:main}
Assume $\bf P1), P2)$ and let $\r_N, \r^2_N$ be defined as in (\ref{eq:Reg-prior1}), (\ref{eq:Reg-prior2}) and discussion around. Then
\begin{itemize}
\item[i)] $HM_{\s}\xleftrightarrow{LT} HM_{\r}$:
\be\label{eq:mainHM}
A_\r(\theta,\a,\b)=\a \sup_{R>0}\left( \a^{-1}A_{\s}(R,\b) -\frac{R^2}{2\theta}+\log R- r(R) -\frac12(\log\theta-1) \right)\,;
\ee
\be\label{eq:mainHM-dual}
A_{\s}(R,\b)=\a \inf_{\theta>0}\left(\frac{R^2}{2\theta}+\a^{-1}A_\r(\theta,\a,\b) -\log R+ r(R)+\frac12(\log\theta-1)\right)\,.
\ee
\item[ii)] $RBM_{\s\s}\xleftrightarrow{LT} RBM_{\s,\g}$:
\be\label{eq:mainRBM1}
A_{\s,\g}(\theta,R_2,\a,\b)=\a\sup_{R_1>0}\left(\a^{-1}A_{\s\s}(R_1,R_2,\b)-\frac{R_1^2}{2\theta}+\log R_1 -\frac12(\log\theta-1)\right)\,;
\ee
\be\label{eq:mainRBM1-dual}
A_{\s\s}(R_1,R_2,\b)=\a\inf_{\theta>0}\left(\frac{R_1^2}{2\theta}+\a^{-1}A_{\s,\g}(\theta,R_2,\a,\b)-\log R_1+\frac12(\log\theta-1)\right)\,.
\ee
\item[iii)] $RBM_{\s\s}\xleftrightarrow{LT} RBM_{\r^2}$:
\be\label{eq:mainRBM2}
A_{\r^2}(\theta,\a,\b)=\sup_{R_1,R_2>0}\left(A_{\s\s}(R_1,R_2)-\frac{\a R_1^2+(1-\a)R_2^2}{2\theta}+\log(R_1^\a R_2^{1-\a})-\sqrt{\a(1-\a)}r(R_1,R_2)-\frac12(\log\theta-1)\right)\,;
\ee
\be\label{eq:mainRBM2-dual}
A(R_1,R_2)=\inf_{\theta_1,\theta_2>0}\left(\frac{\a R_1^2}{2\theta_1}+\frac{(1-\a)R_2^2}{2\theta_2}+A_{\r^2}(\theta_1,\theta_2,\a,\b)-\log(R_1^\a R_2^{1-\a})+\sqrt{\a(1-\a)}r(R_1,R_2)+\frac12(\log\theta_1^\a\theta_2^{1-\a}-1)\right)\,. 
\ee
\end{itemize}

\end{theorem}

The other equivalences of the scheme above can be derived by combining $i), ii), iii)$ and marginalisations. Albeit we will not include that in this paper, these dualities can be established also for the Gibbs distributions (\ref{eq:Gibbs}). 

\

\subsection{Related literature}

Once the model has been properly regularised, spherical, Gaussian or sub-Gaussian priors are all equivalent; so we shall speak indistinctly of Gaussian BMs in what follows.  

The use of Gaussian visible variables is useful to handle real data and has been suggested since the beginning of the theory \cite{hopG}. However the learning and retrieval capabilities of the fully visible BM with Gaussian units are not as good as its $\pm1 $ counterpart at low-load \cite{lungo}, and at high-load they are totally useless \cite{perez,lungo}.  
Restricted architectures are more interesting. RBMs with Gaussian visible and latent variables have been used for instance for factor analysis \cite{ICC, prodgauss} and collaborative filtering \cite{salamino}. In general training, through e.g. contrastive divergence, is slower than that  in a Bernoulli-Gaussian machine \cite{guide,gaussCD,karakida,decelle} and also retrieval is less pronounced \cite{lungo, RBMs}. 

Independently on their performances, Gaussian BMs are of great theoretical relevance from the viewpoint of spin glasses, since their simpler mathematical structure helps our understanding of the, much more complicated,  discrete models. Previous results on the model have been obtained in \cite{AC} and \cite{baik1}. In \cite{baik1} the authors achieve the same result as in our Theorem \ref{Th:Bip-sfer} and prove much more: fluctuations of the free energy are shown to be Gaussian in high temperature and Tracy-Widom for low temperature. The assumptions on the weight distribution are more general then ours, as they only require finiteness of the moments. Yet the method there employed is a sophisticated and technical random matrix argument and our approach is certainly lighter and more accessible to non-specialists.
In \cite{AC} a variational principle for the free energy has been proven only for small $\b$, by means of the so-called Latala method. One great merit of the approach of this paper is to provide a clear interpretation of the replica symmetric nature of the variational formula for the free energy (formulated in terms of the overlap), which is absent in \cite{baik1} and in the present work, even though by a direct comparison with \cite{legendre} one can see that our Lagrange multipliers ($a,b$ below) are essentially shifted overlaps.  In any case it is remarkable that the free energy of a doubly spherical RBM satisfies a fully convex minimisation principle, which is a crucial difference compared to the $\min\max$ of Gauss-Bernoulli \cite{bgg} and Bernoulli-Bernoulli RBMs \cite{Bip}. The reason for that eludes our current understanding and we must defer the discussion of this point to future works. 


\section{Free energy of spherical models}\label{section:sferico} 

In this section we study the HM Hamiltonian (\ref{eq:HM}) and RBM Hamiltonian (\ref{eq:RBM}) with the spherical prior in (\ref{eq:prior}).  Our main results are
\begin{theorem}\label{Th:HM-sfer}
Let $A_{N_1,N_2}$ denote the free energy of $HM_{\s}$. It holds $P$-a.s.
\be\label{eq:AHMsfer}
\lim_{N_1,N_2\to\infty} A_{N_1,N_2} = \a\min_{2q\geq \b(1-\a)\l_+}\left(q R^2 -\frac 12\int\r_{MP}(\l;\a)\log(2q-\b(1-\a)\l) d\l-\log R-\frac12\right)\,.
\ee
\end{theorem}

\begin{theorem}\label{Th:Bip-sfer}
Let $A_{N_1,N_2}$ denote the free energy of $RBM_{\s\s}$ and $\a\leq \frac 12$. It holds $P$-a.s.
\begin{eqnarray}
\lim_{N_1,N_2\to\infty} A_{N_1,N_2} &=&\min_{ab\geq \b^2(1-\a)\l_+}\left(\frac{R^2}{2}\left(\a a + (1-\a) b\right) -\frac 1 2 (1-2\a)\log(b)\right. \nonumber \\
 &-&\left.\frac \a2\int\r_{MP}(\l;\a)\log(ab-\b^2(1-\a)\l) d\l-\log R-\frac12\right) \label{eq:Bip-sfer}\,.
\end{eqnarray}
\end{theorem}

The choice $\a\leq\frac12$ is just a matter of convenience as it will be clear that pre-choosing the largest layer simplifies a lot the notations. For $\alpha\geq \frac12$ one should bear in mind that the Marchenko-Pastur distribution (\ref{eq:MP}) has an atom in zero. 

First of all we prove that the spherical shells constitute a good approximation of the spherical prior. 
We recall that everywhere the quantities with $\widehat \cdot$ are always referred to spherical shell priors. 

\begin{lemma}\label{lemma-gusci}
Let $\widehat A_{N_1,N_2,\e}$ denote the free energy of $HM_{\s_\e}$, $RBM_{\s_\e\s}$ or $RBM_{\s_\e\s_\e}$. Then it is
$$
\lim_{\e\to0}\lim_{N_1,N_2} \widehat A_{N_1,N_2,\e}=\lim_{N_1,N_2}\lim_{\e\to0} \widehat A_{N_1,N_2,\e}\,. 
$$
\end{lemma}

The existence of the limit in the r.h.s. will be proven below. 

\begin{proof}
In this proof we write $N$ to mean $N_1$ or $N_2$ and by $A^{\s}(R; \b)$ the free energy of a spherical model of radius $R$, according to the context. No further details are needed. By Fubini and the mean value theorem there is a $R^*_\e\in[R\sqrt N-\e,R\sqrt N+\e]$ for which
$$
\widehat A_{N_1,N_2}=A_{N_1,N_2}^{\s}(R^*_\e; \b)\,. 
$$
By a simple change of variables we have
$$
A_{N_1,N_2}^{\s}(R^*_\e; \b)=A_{N_1,N_2}^{\s}(R; c_{N,\e}\b)\,, 
$$
with $c_{N,\e}:=R^2_\e/R^2N$. Therefore by Lipschitz continuity of the free energy w.r.t. $\b$
$$
|A^\s_{N_1,N_2}-\widehat A_{N_1,N_2,\e}|\leq \frac\e N\,. 
$$
This, combined with $A^\s_{N_1,N_2}=\lim_{\e\to0} \widehat A_{N_1,N_2,\e}$, gives the assert.
\end{proof}

We first deal with the Hopfield model. Everywhere from now on partition function and pressure will be referred to this model, unless otherwise specified.

The best advantage of the spherical prior is that one can diagonalise the energy (\ref{eq:HM}):
$$
H(x)=-(1-\a)\sum_{i=1}^{N_1}\l_ix_i^2\,.
$$
Thanks to {\bf P2)}, we can restrict our analysis to disorder realisations with spectrum contained in $(-\infty,\l_+]$. More precisely 

\begin{lemma}\label{lemma:HM-sfer}
For any $a>\l_+$ there is a $c>0$ such that
\be\label{eq:lemma-HM-sfer}
E[A_{N_1,N_2}(\b,R)1_{\{\|\frac{\Xi\Xi^T}{N_2}\|_{op}>a\}}]=\OOO{e^{-cN_1}}\,. 
\ee
\end{lemma}

The proof of that follows essentially the same lines of \cite{legendre}. We omit here the details. Also, the next proof is a straightforward adaption of \cite[Proof of (9)]{legendre}. We give it completely in order to introduce the argument used to prove Theorem \ref{Th:Bip-sfer}. 

\begin{proof}[Proof of Theorem \ref{Th:HM-sfer}]
Let $2q>\b(1-\a)\l_+$. Then using (\ref{eq:shells-surface}) we have 
\bea
\e \hat{Z}_{N_1,N_2,\e}(\b,\a,R)&=& \e\ e^{q R^2 N_1}\int_{\R^{N_1}}\s_{N_1,R,\e}(dz)e^{-\sum_i^{N_1}(q-\frac{\b(1-\a)}{2}\l_i)z_i^2}\label{eq:reprZ-Hop}\\
&\leq&e^{q R^2 N_1}\frac{(2\pi)^{\frac {N_1}2}}{|S_{R\sqrt{N_1}}|}\int_{\R^N}\frac{dz}{(2\pi)^{\frac {N_1}2}} e^{-\sum_i^{N_1}(2q-\b(1-\a)\l_i)z_i^2/2}\nn\\
&=&e^{q R^2 N_1}\frac{(2\pi)^{\frac {N_1}2}}{|S_{R\sqrt{N}}|}e^{-\frac 12\sum_i^N\log(2q-\b(1-\a)\l_i)}\nn\,,
\eea
therefore, since $\frac 1{N_1} \log \left(|S_{R\sqrt{N_1}}| /\sqrt{2\pi}^{N_1}  \right) \to \log R +\frac 1 2$ and thanks to Lemma \ref{lemma-gusci} 
$$
\lim\sup_{N_1,N_2}\frac1{N_1+N_2} \log Z_{N_1,N_2}(\b,\a,R)\leq \a q R^2 -\frac \a2\int\r_{MP}(\l;\a)\log(2q-\b(1-\a)\l)-\a\log R-\frac\a2=:\tilde A(q)\,,
$$
whence, as $\tilde A(q)$ is continuous, 
$$
\lim\sup_{N_1,N_2}\frac{1}{N_1+N_2} \log Z_{N_1,N_2}(\b,\a,R)\leq\min_{2q\geq\b(1-\a)\l_+}\tilde A(q).
$$ 
Moreover for $2q>\b(1-\a)\l_+$
$$
\partial_q^2 \tilde A(q)=2\a\int d\l\frac{\r_{MP}(\l;\a)}{(2q-\b(1-\a)\l)^2}>0,
$$
thus $\tilde A(q)$ is  convex and the minimum is attained in a unique point $\bar q$. 
 
\ni Now the reverse bound. Let $\e>0$ and $S^c_{N_1,R,\e}$ the complementary set of the shell $S_{N_1,R,\e}$. It holds
$$
\e \widehat Z_{N_1,N_2,\e}=e^{ q R^2 N_1}\frac{(2\pi)^{\frac {N_1}{2}}}{|S_{R\sqrt{N_1}}|}\int_{\R^{N_1}}\frac{dz}{(2\pi)^{\frac {N_1}{2}}} e^{-\sum_i^{N_1}(q-\frac{\b(1-\a)}{2}\l_i)z_i^2}-e^{ q R^2 N_1}\frac{(2\pi)^{\frac N2}}{|S_{R\sqrt{N_1}}|}\int_{S_{N_1,R,\e}^c}\frac{dz}{(2\pi)^{\frac {N_1}2}} e^{-\sum_i^{N_1}(q-\frac{\b(1-\a)}{2}\l_i)z_i^2}.
$$
For any $\eta>0$ small enough we have
\bea
\int_{S_{N_1,R,\e}^c}\frac{dz}{(2\pi)^{\frac {N_1}2}} e^{-\sum_i^{N_1}(q-\frac{\b(1-\a)}{2}\l_i)z_i^2}&\leq&\exp\left[N_1\left(\eta\left(R^2-\frac\e {N_1}\right)-\frac{1}{2N_1}\sum_j\log(2q-\b\l_j+2\eta)\right)\right]\nn\\
&+&\exp\left[N_1\left(-\eta\left(R^2+\frac\e {N_1}\right)-\frac{1}{2N_1}\sum_j\log(2q-\b\l_j-2\eta)\right)\right]\nn.
\eea
\ni Note now that the r.h.s. is $o(e^{-N})$ if $\frac\e N\to\infty$ as $N\to\infty$. So the greatest contribution is at the scale $\e=\tilde\e N$, $\tilde\e>0$ independent on $N_1$. Therefore
\be\label{eq:max( , , )}
\lim\inf_{N_1,N_2}\widehat A_{N_1,N_2,\tilde\e}\geq\max\left(\tilde A(q), A^1_{\tilde\e}(\eta; q), A^2_{\tilde\e}(\eta; q)\right),
\ee
with
\bea
A^1_{\tilde\e}(\eta; q)&=&\a(q+\eta)R^2-\a\tilde\e\eta-\frac\a2\int d\l\r_{MP}(\l;\a)\log(2(q+\eta)-\b\l)-\a\log R-\frac\a2;\\
A^2_{\tilde\e}(\eta; q)&=&\a(q-\eta)R^2-\a\tilde\e\eta-\frac\a2\int d\l\r_{MP}(\l;\a)\log(2(q-\eta)-\b\l)-\a\log R-\frac\a2.
\eea
Now we show that if $\bar q$ is the unique minimiser of $\tilde A(q)$ it is for $\eta$ small enough
$$
\tilde A(\bar q)=\max\left(\tilde A(\bar q), A^1_{\tilde\e}(\eta; \bar q), A^2_{\tilde\e}(\eta; \bar q)\right)\,,
$$
which will conclude the proof.

To do so we introduce
\bea
\D_-(q;\eta)&:=&\tilde A(q)-A^1_{\tilde\e}(\eta; q)= -\a\eta(R^2-\tilde\e)+\frac\a2\int d\l\r_{MP}(\l;\a)\log\left(\frac{2(q+\eta)-\b\l}{2q-\b\l}\right);\\
\D_+(q;\eta)&:=&\tilde A(q)-A^2_{\tilde\e}(\eta; q)= \a\eta(R^2+\tilde\e)+\frac\a2\int d\l\r_{MP}(\l;\a)\log\left(\frac{2(q-\eta)-\b\l}{2q-\b\l}\right)\,.
\eea
\ni As a function of $\eta$, $\D_\pm(q;\eta)$ are continuous and differentiable, vanishing in $\eta=0$ and with $\lim_{\eta\to+\infty}\D_\pm(q,\eta)=\pm\infty$. Moreover $\D_+(q;\eta)$ is uniformly convex and $\D_-(q;\eta)$ uniformly concave. Thus $\D_-(q;\eta)$ assumes a positive maximum iff the derivative in $\eta=0$ is positive, that is
\be\label{eq:upper-bound-epsi1}
0<-(R^2-\tilde\e)+\int d\l\frac{\r_{MP}(\l;\a)}{2q-\b\l}=\tilde \e-\partial_q \tilde A(q)\,.
\ee

\ni Likewise for $\D_+(q;\eta)$ is always positive iff 
$\left.\partial_\eta \D_+(q;\eta)\right|_{\eta=0}\geq0$, \ie
\be\label{eq:upper-bound-epsi2}
0\leq(R^2+\tilde\e)+\int d\l\frac{\r_{MP}(\l;\a)}{2q-\b\l}=\tilde \e+\partial_q \tilde A(q).
\ee
Combining (\ref{eq:upper-bound-epsi1}) and (\ref{eq:upper-bound-epsi2}) we get
$$
-\tilde\e\leq \left.\partial_q \tilde A(q)\right|_{q=\tilde q}<\tilde\e\,,,
$$
that is $\tilde q$ is the unique stationary point of $\tilde A(q)$. With this choice of $q$, relation (\ref{eq:max( , , )}) gives
\be
\lim\inf_{N_1,N_2} \widehat A_{N_1,N_2.\tilde \e}\geq \min_{q\geq\b\l_+}\tilde A(q) \,.
\ee
As $\tilde\e$ can be taken arbitrarily small, we recover (\ref{eq:AHMsfer}).
\end{proof}

Now we pass to compute the free energy of $RBM_{\s\s}$ by adapting the same method as before. Until the end of this section, Hamiltonian, partition function etc will be referred to the $RBM_{\s\s}$.
First step is to use the singular value decomposition of the matrix $\{\xi_{ij}\}$ to write
\be\label{eq:diagRBM}
\frac{1}{\sqrt{N_1+N_2}}(x,\xi y)=\sqrt{1-\a}\sum_{i\in[N_1]}\s_i\tilde x_i\tilde y_i\,,
\ee
where $\tilde x,\tilde y$ are related respectively to $x,y$ by an orthogonal transformation. Note that this decomposition removes automatically $N_2-N_1$ cyclic coordinates of the second layer. 

First of all we show that the part of the spectrum which falls away the support of the Marchenko-Pastur law is negligible for the free energy. The proof is adapted from \cite{legendre} and we will stress only the most salient points of it. 
\begin{lemma}
There exists $c,C>0$ so that $$E[A_{N_1,N_2}1_{\{\max_{i\in[N_1]}\s_i>\bar\s\}}]\leq Ce^{-cN_1}\,,$$ where
$\bar \s:= 1+\sqrt{\frac{\a}{1+\a}}$. 
\end{lemma}
\begin{proof}
For any $a>\bar\s$ we compute
$$
E[e^{\a\b\sum_{i\in[N_1]}\s_i\tilde x_i\tilde y_i}]=E[e^{\b\sum_{i\in[N_1]}\s_i\tilde x_i\tilde y_i}1_{\{\max_{i\in[N_1]}\s_i\leq a\}}]+E[e^{\b\sum_{i\in[N_1]}\s_i\tilde x_i\tilde y_i}1_{\{\max_{i\in[N_1]}\s_i>a\}}]\,.
$$
The first summand is easily estimated by
\be\label{eq:bound1}
E[e^{\b\sum_{i\in[N_1]}\s_i\tilde x_i\tilde y_i}1_{\{\max_{i\in[N_1]}\s_i\leq a\}}]\leq E[e^{\b a\sum_{i\in[N_1]}\tilde x_i\tilde y_i}]\leq E[e^{\b a \|\tilde x\|\|\tilde y\|}]\,. 
\ee
Arguing as in the proof of \cite[Proposition 1]{legendre} we can write for a $C>0$
\be\label{eq:bound2}
E[e^{\b\sum_{i\in[N_1]}\s_i\tilde x_i\tilde y_i}1_{\{\max_{i\in[N_1]}\s_i>a\}}]\leq Ca\sqrt{2\pi (N_1+N_2)}e^{\frac{\b^2}{2(N_1+N_2)}\|\tilde x\|^2\|\tilde y\|^2}\,.
\ee
Then (\ref{eq:bound1}) and (\ref{eq:bound2}) give an annealed bound for the free energy:
\be
\lim\sup_{N_1,N_2} A_N\leq \max\left(\bar\s\b\sqrt{\a(1-\a)}R^2\,,\,\,\frac12 \b^2\a(1-\a)R^2_1R^2_2 \right)<\infty\,. 
\ee

With this bound at hand, we repeat mutatis mutandis the steps of the proof of \cite[Lemma 1]{legendre} to get
\be
E[A_{N_1,N_2}1_{\{\max_{i\in[N_1]}\s_i>\bar\s\}}]\leq\sqrt{\a(1-\a)\b^2R^4+a^2}\sqrt{P(\|\Xi\Xi^T/N_2-\mathbb I\|_{op}\geq t+\sqrt{\l_+}-1)}
\ee
and the assert follows from (\ref{eq:impo-spectrum}). 
\end{proof}

\begin{proof}[Proof of Theorem \ref{Th:Bip-sfer}]

Now  start from (\ref{eq:diagRBM}) and write for any $a,b$ such that $ab>\b^2(1-\a)\l_+$
\bea
Z_{N_1,N_2}&=&\int\s_1(dx)\s_2(dy) \exp\left( \sqrt{1-\a} \b\sum_{i\in[N_1]}x_iy_i\s_i-\frac a2 \|x\|^2-\frac b2 \|y\|^2+\frac a2 R^2N_1+\frac b2 R^2N_2\right)\nn\\
&=&\exp\left(\frac {R^2}{2} (aN_1+b N_2)\right) \int\s_1(dx)\s_2(dy)\exp\left( -\frac 1 2 \sum_{i=1}^{N_2} (z_i^T,M_i z_i)\right)\label{eq:repr-Z-RBM}\,,
\eea
where we have defined 
$$
z^T_i := \begin{cases}
(x_i,y_i)&i=1,\ldots , N_1\,,\\
(0,y_i)&i=N_1+1,\ldots , N_2
\end{cases}$$
and the $2\times 2$ symmetric positively defined matrix 
$$
M_i :=
  \left( {\begin{array}{cc}
   a & -\b\sqrt{1-\a}\s_i \\
   -\b\sqrt{1-\a}\s_i & b \\
\end{array} } \right)\,,\quad i\in[N_1]\qquad M_i:=\diag(0,b)\,,\quad i=N_1+1,\ldots , N_2\,.
$$ 
Thus by (\ref{eq:shells-surface})
\bea
\e^2\hat{Z}_{N_1,N_2,\e}&\leq& \exp\left(\frac {R^2}{2} (aN_1+b N_2)\right) \frac{(2\pi)^{\frac{N_1}{2}} (2\pi)^{\frac{N_2}{2}} }{S^R_{N_1} S^R_{N_2}} \int \frac{dx dy}{\sqrt{2\pi}^N}\exp\left( -\frac 12\sum_{i=1}^{N_2} (z_i^T,M_i z_i)\right),\nn\\ 
&=& \exp\left(\frac {R^2}{2} (aN_1+b N_2)\right) \frac{(2\pi)^{\frac{N_1}{2}} (2\pi)^{\frac{N_2}{2}} }{S^R_{N_1} S^R_{N_2}} b^{-\frac{N_2-N_1}{2}} \prod_{i=1}^{N_1}(ab-\b^2(1-\a)\s^2_i)^{-\frac 12}\,.
\eea
We introduce
\be\label{eq:RBM-trial}
\bar A(a,b):= \frac {R^2}{2} (a\a+b(1-\a))-\log R-\frac12-\frac{1-2\a}{2}\log b -\frac{\a}{2}\int \r_{MP}(d\l;\a)\log(ab-(1-\a)\b^2\l)\,.
\ee
and 
$$
\lim\sup_{N_1,N_2}A_{N_1,N_2}\leq \bar A(a,b)\,,\quad \forall a,b>0\,:\,\,ab>(1-\a)\b^2\l_+\,. 
$$
A direct inspection shows that $A(a,b)$ is jointly uniformly convex for any $\a\in[0,1/2]$ if $ab>(1-\a)\b^2\l_+$, therefore
\be\label{eq:upperbound.RBM}
\lim\sup_{N_1,N_2}A_{N_1,N_2}\leq \min_{ab>(1-\a)\b^2\l_+} \bar A(a,b)\,. 
\ee

We record for later use the gradient coordinates 
\bea
\partial_a \bar A(a,b)&=&\frac{\a R^2}{2}-\frac\a2\int \r_{MP}(d\l;\a)\frac{b}{ab- (1-\a)\b^2\l}\,;\label{eq:grad1RBM}\\
\partial_b \bar A(a,b)&=&\frac{(1-\a)R^2}{2} -\frac{1-2\a}{2b}-\frac\a2\int \r_{MP}(d\l;\a)\frac{a}{ab- (1-\a)\b^2\l}\,.\label{eq:grad2RBM}
\eea


For the reverse bound consider again  spherical shells around $S_{R,N_1}$ and $S_{R,N_2}$ of thickness $\e$, that we name respectively $S_{1,\e}$ and $S_{2,\e}$. We split $S_{1,\e}=\R^{N_1}\setminus S^c_{1,\e}$ and $S_{2,\e}=\R^{N_2}\setminus S^c_{2,\e}$ and set $S_{\e}:=S_{1,\e}\times S_{2,\e}$. We have

\bea
\e^2 \widehat Z_{N_1,N_2,\e}&:=&\e^2\int\s_{N_1, R,\e}(dx)\s_{N_2,R,\e}(dy)e^{-\b H_N(x,y)}\nn\\
&=&\exp\left(\frac {R^2}{2} (aN_1+b N_2)\right) \int_{\R^{N_1}\times \R^{N_2}} \frac{dxdy}{|S_1||S_2|}\exp\left( -\frac 1 2 \sum_{i=1}^{N_2} (z_i^T,M_i z_i)\right)\nn\\
&-&\exp\left(\frac {R^2}{2} (aN_1+b N_2)\right) \int_{S^c_{\e}} \frac{dxdy}{|S_1||S_2|}\exp\left( -\frac 1 2 \sum_{i=1}^{N_2} (z_i^T,M_i z_i)\right)\label{eq:II-termRBM}\,,
\eea
where we used again the representation (\ref{eq:repr-Z-RBM}). The free energy associated to the first summand was already computed above in the thermodynamic limit. Therefore we have to upper bound the second summand. We consider four contributions according to the following decomposition. Let $\kappa\in\{-1,1\}$ and put
\bea
S_\e^{\kappa,1}&:=&\{x\in\R^{N_1}\,, y\in\R^{N_2}\,:\,\, \kappa(\|x\|^2-R^2N_1)\geq \e N_1\,\}\,,\nn\\
S_\e^{\kappa,2}&:=&\{x\in\R^{N_1}\,, y\in\R^{N_2}\,:\,\, \kappa(\|y\|^2-R^2N_2)\geq \e N_2\,\}\,.\nn 
\eea
Thus $S_\e^c=\bigcup_{\kappa\in\{-1,1\},j\in\{1,2\}}S_\e^{\kappa,j}$. Moreover we pick $\eta>0$ small enough and set
\bea
M^{(\kappa,j)}_i(\eta) &:=&
  \left( {\begin{array}{cc}
   a+(2-j)\kappa\eta & -\b\sqrt{1-\a}\s_i \\
   -\b\sqrt{1-\a}\s_i & b+(j-1)\kappa\eta \\
\end{array} } \right)\,,\quad i\in[N_1]\nn\\ 
M^{(\kappa,j)}_i(\eta)&:=&\diag(0,b+(j-1)\kappa\eta)\,,\quad i=N_1+1,\ldots , N_2\,.\nn
\eea



Also, we put
\bea
Z_{N_1,N_2,\e}^{\kappa,j}&:=& \exp\left(\frac {R^2}{2} (aN_1+b N_2)+\frac{(2-j)N_1\kappa\eta a}{2}+\frac{(j-1)N_2\kappa\eta b}{2}-\eta \e N_j\right)\nn\\
&&\int_{S_\e^{\kappa,j}} \frac{dxdy}{|S_1||S_2|}\exp\left( -\frac 1 2 \sum_{i=1}^{N_2} (z_i^T,M^{(\kappa,j)}_i(\eta) z_i)\right)\nn\\
\eea
so that 
$$
(\ref{eq:II-termRBM})=-\sum_{\kappa\in\{-1,1\},j=1,2} Z_{N_1,N_2,\e}^{\kappa,j}\,. 
$$
We conclude that for any $\eta>0$ sufficiently small and $a,b$ with $ab>(1-\a)\b^2\l_+$ 
\be
\lim\sup_{N_1,N_2} \widehat A_{N_1,N_2,\e}\geq \max(\bar A(a,b), \{A_\e^{\kappa, j}(\eta; a,b)\}_{\kappa\in\{-1,1\},j\in\{1,2\}})\,,
\ee
where 
\bea
A_\e^{\kappa,j}(\eta; a,b)&:=&-\eta\e+\frac12\a a(R^2+(2-j)\kappa\eta)+\frac12(1-\a) b(R^2+(j-1)\kappa\eta)\nn\\
&-&\log R -\frac12-\frac{1-2\a}{2}\log (b+(j-1)\kappa\eta)\nn\\
&-&\frac\a2 \int \r_{MP}(d\l;\a)\log((a+(2-j)\kappa\eta)(b+(j-1)\kappa\eta)- (1-\a)\b^2\l)\,. \nn
\eea

In analogy with the proof of Theorem \ref{Th:HM-sfer} we define 
\bea
\D_{\e}^{\kappa,j}(\eta)&:=&\bar A(a,b)-A_\e^{\kappa,j}(\eta; a,b)\nn\\
&=&\eta \e-(2-j)\kappa\frac{\eta\a a}{2}-(j-1)\kappa\frac{\eta(1-\a) b}{2}+\frac{1-2\a}{2}\log(b+(j-1)\kappa\eta)\nn\\
&+&\frac\a2 \int \r_{MP}(d\l;\a)\log\left(\frac{(a+(2-j)\kappa\eta)(b+(j-1)\kappa\eta)- (1-\a)\b^2\l}{ab-(1-\a)\b^2\l}\right)\,.\nn
\eea
We need to show that $\D_{\e}^{\kappa,j}(\eta)\geq0$ for $\eta>0$ small. As before, to do so it suffices to prove the derivative in the origin to be non-negative uniformly in $\e>0$. Bearing in mind (\ref{eq:grad1RBM}), (\ref{eq:grad2RBM}) we have
\bea
\frac{d}{d\eta} \D_{\e}^{\kappa,j}(\eta)\big|_{\eta=0}&=&\e+(2-j)\kappa\left(-\frac{a\a R^2}{2}+\frac\a2 \int \r_{MP}(d\l;\a)\left(\frac{b}{ab-(1-\a)\b^2\l}\right)\right)\nn\\
&+&(j-1)\kappa\left(-\frac{b(1-\a) R^2}{2}+\frac{1-2\a}{2b}+\frac\a2 \int \r_{MP}(d\l;\a)\left(\frac{a}{ab-(1-\a)\b^2\l}\right)\right)\nn\\
&=&\e-(2-j)\kappa\partial_a\bar A-(j-1)\kappa\partial_b\bar A\geq0\,. \nn
\eea
Since the inequality must hold for any $\e>0$, $\kappa\in\{-1,1\}$ and $j\in\{1,2\}$, we have to pick $(\bar a, \bar b)=\arg\min \bar A$. Therefore 
$$
\lim\sup_{N_1,N_2} A_{N_1,N_2,\e}\geq \max(\bar A(\bar a,\bar b), \{A^{\e,\kappa_1,\kappa_2}(\eta; \bar a,\bar b)\}_{\kappa_1,\kappa_2\in\{-1,1\}})= \min_{ab>\b^2(1-\a)\l_+}\bar A( a, b)\quad \forall \e>0\,,
$$
which combined with (\ref{eq:upperbound.RBM}) proves the theorem. 
\end{proof}

\section{Legendre equivalences of priors}\label{section:legendre}

In this section we explain the Legendre equivalence of spherical models on general terms.
First of all we prove some a priori estimates ensuring the boundedness of the free energy in the thermodynamic limit. This will be used to cut the tails of the Gaussian distributions of the prior. 

The first quick remark is that combining Theorem \ref{Th:HM-sfer} and a marginalisation we have
 \begin{corollary}
Let $A_N$ be the free energy of $RBM_{\s,\g}$. Then $\lim_NA_N$ exists $P$-a.s.
 \end{corollary}

Next we focus on the Hopfield model with Gaussian prior previously defined.

\begin{lemma}\label{lemma:apriori-bound-HM}
Let $A_{N_1,N_2}$ be the free energy of $HM_{\r}$. There is $f(\l_+,\b)$ continuous and bounded for which
\be\label{eq:ann-bound}
\lim\sup_{N_1,N_2} E[A_{N_1,N_2}]\leq f(\l_+,\b)\,. 
\ee
\end{lemma}
\begin{proof}
Let $a>\l_+$ and set for brevity $\l^*:=\max_{i\in[N_1]}\l_i$. We write
$$
\frac{1}{N_1+N_2}E[\log Z_{N_1,N_2}]=\sum_{k\geq0}E\left[\frac{1}{N_1+N_2}\log Z_{N_1,N_2}\,|\,\,(k+1)a\geq \l^*> ka\right ]P((k+1)a\geq \l^*> ka)\,. 
$$
Inside the conditional expectation we can bound
$
H_N(z)\leq a\|z\|^2\, 
$. Therefore
\bea
E\left[\frac{1}{N_1+N_2} \log Z_{N_1,N_2}\,|\,\,(k+1)a\geq \l^*> ka\right]&\leq& \frac{1}{N_1+N_2} \log\int\g_{N_1,\theta}(dx) \exp\left(\|x\|^2(\b a(k+1))-Nr\left(\frac{\|x\|}{\sqrt N}\right)\right)\nn\\
&\leq&\frac{1}{N_1+N_2}\max_{R\geq0}\left(R^2(\b a(k+1))-N_1r\left(\frac{R}{\sqrt N_1}\right)\right)\,\nn\\
&=&\max_{R\geq0}\left(R^2(\b a(k+1))-r\left(R\right)\right)\,.\nn
\eea
On the other hand by assumption
\be
P((k+1)a\geq \l^*> ka)\leq 2e^{-ca^2k^2N_1}\,,\quad c>0\,. 
\ee
In conclusion
\be
E[A_{N_1,N_2}]\leq \sum_{k\geq0}e^{-ca^2k^2N_1}\max_{R\geq0}\left(R^2(\b a(k+1))-r\left(R\right)\right)=:f(a,\b)\,,
\ee
a continuous bounded function. In particular the estimate holds also for $a\to\l_+$. 
\end{proof}

The analogue statement for $RBM_{\r^2}$:
\begin{lemma}\label{lemma:apriori-bound-HM}
Let $A_{N_1,N_2}$ be the free energy of $RBM_{\r^2}$. There is $f(\l_+,\b)$ continuous and bounded for which
\be\label{eq:ann-bound}
\lim\sup_{N_1,N_2} E[A_{N_1,N_2}]\leq f(\l_+,\b)\,. 
\ee
\end{lemma}
\begin{proof}
Same proof as before, noting
\bea
\b(1-\a)\sum_{i}\s_ix_iy_i-\sqrt{N_1N_2}r\left(\frac{\|x\|}{\sqrt {N_1}}, \frac{\|y\|}{\sqrt {N_2}}\right)&\leq& \b(1-\a)a\|x\|\|y\|-\sqrt{N_1N_2}r\left(\frac{\|x\|}{\sqrt {N_1}}\right)\nn\\
&\leq&\max_{R_1,R_2}\left(\b(1-\a)aR_1R_2-r(R_1,R_2)\right)\,\nn
\eea
for any $a>\s_+$. 
\end{proof}

The above results immediately allow us to achieve the following useful lemma. 

\begin{lemma}\label{lemma:coda}
Let $R,\d>0$, $N\in\N$. It holds for some $C,c>0$
\bea
\int_{\{\|x\|^2\geq R^2N_1^{1+\d}\}} \r_{N_1}(dx)e^{-\b H_{N_1,N_2}(x)}&\leq& C e^{-cN_1^{1+\d}}\,,\label{eq:coda1}\\
\int_{\{\|x\|^2\geq R^2N_1^{1+\d}\}} \g_{N_1,\theta}(dx)\s_{R,N_2}(dy)e^{-\b H_{N_1,N_2}(x,y)}&\leq& C e^{-cN_1^{1+\d}}\,,\label{eq:coda2}\\
\int_{\{\|x\|^2\geq R^2N_1^{1+\d}\}\cup \{\|y\|^2\geq R^2N_2^{1+\d}\}} \r_{N_1,N_2}^2(dxdy)e^{-\b H_{N_1,N_2}(x,y)}&\leq& C e^{-cN_1^{1+\d}}\label{eq:coda3}\,.
\eea 
\end{lemma}
\begin{proof}
We prove only (\ref{eq:coda1}), as (\ref{eq:coda2}), (\ref{eq:coda3}) are similar. Let us write
\bea
\int_{\{\|x\|^2\geq R^2N_1^{1+\d}\}} \r(dx)e^{-\b H_{N_1,N_2}(x)}&\leq& e^{-\frac{R^2N_1^{1+\d}}{2\theta}}\int_{\{\|x\|^2\geq R^2N_1^{1+\d}\}} \g_{N_1,2\theta}(dx)e^{-\b H_{N_1,N_2}(x)-N_1r\left(\frac{\|x\|}{\sqrt N_1}\right)}\nn\\
&\leq&e^{-\frac{R^2N_1^{1+\d}}{2\theta}} Z_{N_1,N_2}(2\theta) =\exp\left( (N_1+N_2)A_{N_1,N_2}(2\theta)-\frac{R^2N_1^{1+\d}}{2\theta}\right)\,. \nn
\eea
Here we have emphasised the dependence on $\theta$ of partition function and free energy. Then (\ref{eq:coda1}) follows from Lemma \ref{lemma:apriori-bound-HM}.
\end{proof}

Now we are ready to prove the Legendre equivalences of Theorem \ref{th:main}. We shall prove only (\ref{eq:mainHM}), (\ref{eq:mainRBM1}), (\ref{eq:mainRBM2}); the dual relations (\ref{eq:mainHM-dual}), (\ref{eq:mainRBM1-dual}), (\ref{eq:mainRBM2-dual}) then follow directly, as one can easily verify the inverse Legendre transformation to be also well defined and involutive. 

We start by $i)$, where we deal with a single Gaussian prior. 
Let $\e>0$, $\d>0$. From now on we will systematically omit the dependence on $\d$ of the objects we will operate with. Let further $r<N_1^{\d}/\e$, $R_0:=0$, $R_{r+1}:=N_1^{\d}$, $\{R_i\}_{i=1,\ldots,r}\subset[0,N_1^{\d})$ with $|R_{i+1}-R_i|<2\e$, and decompose
$
\R^{N_1}:=\bigcup_{i=0}^rS^{[i]}_{N_1,\e}\cup T
$, where 
$$
S^{[i]}_{N_1,\e}:=\{z\in\R^{N_1}\,\:\, R_i\sqrt{N_1}\leq \|z\|_2\leq R_{i+1}\sqrt{N_1}\}\,,\quad T:=\{z\in\R^{N_1}\,\:\, \|z\|_2\geq {N_1}^{\d+\frac12}\}\,.
$$
Comparing with (\ref{eq:Sfshell}) one easily sees that the $S^{[i]}_{\e,N_1}$ are spherical shells. We denote by $\s^{[i]}_{N_1,\e}$ the uniform distributions on  these shells. 
Then we have
\be\label{eq:dec}
Z^{\r}_{N_1,N_1}=\sum_{i=0}^r Z^{\r}_{N_1,N_2}[i]+\tilde Z^{\r}_{N_1,N_2}\,,
\ee
where
$$
Z^{\r}_{N_1,N_2}[i]:=\int_{S^{[i]}_{N_1,\e}}\r_{N_1}(dx)e^{-\b H_{N_1,N_2}(x)}\,,\quad \tilde Z^{\r}_{N_1,N_2}:=\int_{T}\r_{N_1}(dx)e^{-\b H_{N_1,N_2}(x)}\,.
$$
The tail term $\tilde Z^\r_{N_1,N_2}$ gives a negligible contribution thanks to Lemma \ref{lemma:coda} and we will ignore it all the time.
Thus by (\ref{eq:dec}) we get 
$$
\max_{i\in[r]} Z^{\r}_{N_1,N_2}[i]\leq Z^{\r}_{N_1,N_2,\b}\leq \frac{N^{\d}}{\e}\max_{i\in[r]} Z^{\r}_{N_1,N_2}[i]\,.
$$
Therefore setting
\be\label{eq:def-tildeA}
A_{N_1,N_2}[i]:=\left(\frac{1}{N_1+N_2}\log Z^{\r}_{N_1,N_2}[i]\right)
\ee
we have
\be\label{eq:sandwich}
\max_{i\in[r]}\left(A_{N_1,N_2}[i]\right)\leq \frac{1}{N_1+N_2}\log Z^{\r}_{N_1,N_2}\leq \max_{i\in[r]}\left(A_{N_1,N_2}[i]\right)+\frac{\d\log N_1-\log \e}{N_1+N_2}\,. 
\ee
So the free energy of $HM_{\r}$ is given by the limit of $\max_{i\in[r]}\left(A_{N_1,N_2}[i]\right)$, provided it exists.

We notice now that by continuity for any $x\in S^{[i]}_{\e,N_1}$
\be\label{eq:Rtilde}
\frac{\|x\|^2}{2\theta}+N_1r\left(\frac{\|x\|}{\sqrt N_1}\right)=\frac{\tilde R_i^2N_1}{2\theta}+N_1r\left(\tilde R_i\right)+\OOO{\e} 
\ee
where $\tilde R_i\in[R_i, R_{i+1}]$. Therefore
\be
Z^{\r}_{N_1,N_2,\b}[i]=e^{-\frac{\tilde R_i^2N_1}{2\theta}-N_1r\left(\tilde R_i \right)-\OOO{\e}}\frac{|S^{[i]}_{N_1,\e}|}{\sqrt{2\pi\theta}^{N_1}}\int \s^{[i]}_{N_1,\e}(dx)e^{-\b H_{N_1,N_2}(x)}\,
\ee
and
\bea
A_{N_1,N_2}[i]&=&\frac{N_1}{N_1+N_2}\left(\frac{1}{N_1}\log\left(\frac{|S^{[i]}_{\e,N_1}|}{\sqrt{2\pi\theta}^{N_1}}\right)-\frac{\tilde R_i^2}{2\theta }-r\left(\tilde R_i\right)\right)+\widehat A_{N_1,N_2}^{\e}\left(\tilde R_i\right)+\frac{\OOO{\e}}{N_1+N_2}\nn\\
&=:&\tilde A_{N_1,N_2,\e}\left(\tilde R_i\right)+\frac{\OOO{\e}}{N_1+N_2}\nn\,.
\eea
Thus we see that the existence of the thermodynamic limit is ensured by Lemma \ref{lemma-gusci} and we have
$$
\lim_{N_1,N_2}\max_{i\in[r]}\left(A_{N_1,N_2}[i]\right)=\sup_{R>0}\lim_{N_1,N_2}\left(\tilde A_{N_1,N_2,\e}(R)\right)=:\sup_{R>0}\tilde A_\e(R)\,. 
$$
By the uniform concavity of $\tilde A_\e(R)$ we have
$$
\lim_{\e\to0}\sup_{R>0}\tilde A_\e(R)=\sup_{R>0}\lim_{\e\to0}\tilde A_\e(R)=-\frac{\a}{2}\log\theta+\a\sup_{R>0}\left(\log R+\frac12+\a^{-1}A_{\s}(R,\b)-r(R)-\frac{R^2}{2\theta}\right)\,
$$
and the proof of (\ref{eq:mainHM}) is concluded. 

\

Note that in this argument the regularising function $r$ plays essentially no role. So it can be set to zero and repeat verbatim all the previous steps for the $RBM_{\s,\g}$. This way we obtain (\ref{eq:mainRBM1}). 

\

Finally we turn to $RBM_{\r^2}$. Here we have to slice up both Gaussian priors and the previous construction easily extends. We just sketch the argument, stressing only the points in which it differs from above. Let $\e>0$, $\d>0$, $r<N_1^{\d}/\e$, $r'<N_2^{\d}/\e$, $R_0,R'_0:=0$, $R_{r+1}:=N_1^{\d}$, $R'_{r'+1}:=N_2^{\d}$, $\{R_i\}_{i=1,\ldots,r}\subset[0,N_1^{\d})$ with $|R_{i+1}-R_i|<2\e$, $\{R'_i\}_{i=1,\ldots,r'}\subset[0,N_2^{\d})$ with $|R'_{i+1}-R'_i|<2\e$. Decompose
$
\R^{N_1}\times \R^{N_2}:=\bigcup_{i\in[r], j\in[r']}S_{N_1,\e}^{[i]}\times {S'}_{N_2,\e}^{[j]} \cup T
$, where 
\bea
S^{[i]}_{N_1,\e}&:=&\{z\in\R^{N_1}\,\:\, R_i\sqrt{N_1}\leq \|z\|_2\leq R_{i+1}\sqrt{N_1}\}\,,\nn\\
{S'}^{[i]}_{N_2,\e}&:=&\{z\in\R^{N_2}\,\:\, R'_i\sqrt{N_2}\leq \|z\|_2\leq R'_{i+1}\sqrt{N_2}\}\,,\nn\\
T&:=&\{z\in\R^{N_1}\,\:\, \|z\|_2\geq {N_1}^{\d+\frac12}\}\cup \{z\in\R^{N_2}\,\:\, \|z\|_2\geq {N_2}^{\d+\frac12}\}\,.\nn
\eea

Again $T$ can be neglected due to Lemma \ref{lemma:coda}. We have to evaluate
\bea
&&Z^{\r^2}_{N_1,N_2}[i,j]:=\int_{S^{[i]}_{N_1,\e}\times {S'}_{N_2,\e}^{[j]}}\r^2_{N_1N_2}(dxdy)e^{-\b H_{N_1,N_2}(x,y)}\,,\nn\\
&=&\frac{|S^{[i]}_{N_1,\e}\times {S'}_{N_2,\e}^{[j]}|}{\sqrt{2\pi\theta}^{N_1+N_2}} e^{-\frac{\tilde R^2_iN_1+{R'}^2_jN_2}{2\theta}-\sqrt{N_1N_2}r\left(\tilde{R}_i, \tilde{R'}_j\right)+\OOO{\e}} \int_{S^{[i]}_{N_1,\e}\times {S'}_{N_1,\e}^{[j]}}\s^{[i]}_{N_1,\e}(dx)\s^{[i]}_{N_1,\e}(dy) e^{-\b H_{N_1,N_2}(x,y)}\,,\nn
\eea
where the $\tilde R_i$ and $\tilde {R'}_i$ are introduced as before (see (\ref{eq:Rtilde})). Therefore

\bea
A_{N_1,N_2,\b}[i,j]&:=&\frac{1}{N_1+N_2}\log Z^{\r^2}_{N_1,N_2,\b}[i,j]\nn\\
&=&\frac{1}{N_1+N_2}\log\left(\frac{|S^{[i]}_{N_1,\e}\times {S'}_{N_2,\e}^{[j]}|}{\sqrt{2\pi\theta}^{N_1+N_2}}\right)+\frac{\OOO{\e}}{N_1+N_2} -\frac{\tilde R^2_i+{R'}^2_j}{2\theta}\nn\\
&-&\frac{\sqrt{N_1N_2}}{N_1+N_2}r\left(\tilde R_i, \tilde {R'}_j\right)+\widehat A_{N_1,N_2,\e}(\tilde R_i, \tilde R_j)\nn\,,
\eea
where $\widehat A_{N_1,N_2,\e}(\tilde R_i, \tilde {R'}_j)$ is the free energy of the RBM whose priors are the spherical shell measures of centres $\tilde R_i, \tilde {R'}_j$. The argument to pass to the thermodynamic limit is then the same, so (\ref{eq:mainRBM2}) is obtained.


\end{document}